\documentclass{article}
\usepackage{fullpage}
\usepackage[utf8]{inputenc}
\usepackage[T1]{fontenc}
\usepackage{multirow}
\usepackage{amsmath, amssymb, amsthm}
\usepackage{enumerate}
\usepackage[bold,basic]{complexity}
\usepackage{xspace}
\usepackage{authblk}
\usepackage{tikz}

\newcommand{\eps}{\varepsilon}

\newcommand{\DIP}{\textsc{DynIP}\xspace}
\newcommand{\DHD}{\textsc{DynHD}\xspace}
\newcommand{\DaHD}{\textsc{DynApproxHD}\xspace}
\newcommand{\DEM}{\textsc{DynEM}\xspace}
\newcommand{\OMv}{\textsc{OMv}\xspace}

\newcommand{\Time}{\mathcal{T}}
\renewcommand{\polylog}{\mathrm{polylog \;}}

\newtheorem{conjecture}{Conjecture}
\newtheorem{problem}{Problem}
\newtheorem{lemma}{Lemma}
\newtheorem{theorem}{Theorem}
\newtheorem{corollary}{Corollary}

\title{Upper and lower bounds for dynamic data structures on strings}
\author[1]{Raphael Clifford\thanks{\texttt{Raphael.Clifford@bristol.ac.uk}. Supported by EPSRC
fellowship EP\/J019283\/1.}}
\author[2]{Allan Gr\o nlund\thanks{\texttt{jallan@cs.au.dk}. Supported by Center for Massive Data Algorithmics, a Center of the Danish National Research Foundation, grant DNRF84.}}
\author[2]{Kasper Green Larsen\thanks{\texttt{larsen@cs.au.dk}. Supported by a Villum Young Investigator grant, an AUFF starting grant and by Center for Massive Data Algorithmics, a Center of the Danish National Research Foundation, grant DNRF84.}}
\author[3]{Tatiana Starikovskaya\thanks{tat.starikovskaya@gmail.com}}

\affil[1]{University of Bristol, Department of Computer Science, Bristol, U.K.}
\affil[2]{Aarhus University, Department of Computer Science, Aarhus, Denmark}
\affil[3]{École Normale Supérieure, Department of Computer Science, Paris, France}

\begin{document}
\maketitle

\begin{abstract}
We consider a range of simply stated dynamic data structure problems on strings. An update changes one symbol in the input and a query asks us to compute some function of the pattern of length $m$ and a substring of a longer text. We give both conditional and unconditional lower bounds for variants of exact matching with wildcards, inner product, and Hamming distance computation via a sequence of reductions. As an example, we show that there does not exist an $O(m^{1/2-\eps})$ time algorithm for a large range of these problems unless the online Boolean matrix-vector multiplication conjecture is false. We also provide nearly matching upper bounds for most of the problems we consider. 
\end{abstract}

\section{Introduction}
The search for lower bounds  provides one of the greatest challenges in computer science. Progress in finding better truly unconditional lower bounds continues in slow but steady steps. There appears however, in the short term at least, to be no realistic prospect of finding unconditional lower bounds which are polynomial in the size of the input.  One of the most exciting discoveries in recent years has been that such polynomial lower bounds can be given for a range of problems in $\P$ conditional on the hardness of a small set of well known and conjectured to be hard problems~\cite{AW:2014,AWW:2014,Bringmann:2014, BI:SETH:2015, ABW:2015,Bringmann:FOCS:2015}. These include the Strong Exponential Time Hypothesis (SETH), 3-SUM and online Boolean matrix-vector product (OMv).  

In this paper we study the hardness of a number of simply stated dynamic string problems and show both conditional lower bounds based on the OMv conjecture (see Conjecture~\ref{conj:omv_conjecture} for a precise statement) as well as unconditional lower bounds.  We will also give new upper bounds which in many cases will nearly match our new conditional lower bounds. Each problem will have the following form.

\begin{problem}
Consider a text $T$ of length $n$ and a pattern $P$ of length $m$. An update to the pattern (or text) is a pair $(j, \sigma)$ which indicates that the letter at index $j$ in the pattern (or text) is to be substituted with the letter $\sigma$. The task is to develop a dynamic data structure on $P$ and $T$ that supports the following queries: Given a position $i$ of $T$, output $f(P, T[i,\dots, i+m-1])$.   
\end{problem}

Unless stated otherwise, we allow updates to both the pattern $P$ and the text $T$. The different functions $f$ we will consider are Hamming distance (\DHD), inner product (\DIP) and exact matching with wildcards (\DEM). These functions have formed the core of pattern matching with errors and wildcards for many years and have been extensively studied in both the standard offline pattern matching setting and to a lesser extent online and streaming. To the best of our knowledge, this is the first exploration of the  complexity of pattern matching with errors and wildcards as a fully dynamic data structure problem.

By way of preparation, we give $O(\sqrt{m\log{m}})$ query and update times for exact inner product, exact matching with wildcards, and for dynamic Hamming distance over constant-sized alphabets, as well as $O(m^{3/4}\log^{1/4}{m})$-time algorithm for dynamic Hamming distance over polynomial-size alphabets. These algorithms are derived via a lazy rebuilding scheme. We then show in Theorem~\ref{thm:conditional-main} that there does not exist an $O(m^{1/2 - \epsilon})$ time solution to any of these problems unless the online Boolean matrix-vector conjecture is false. The lower bound for dynamic exact matching with wildcards is particularly interesting as it is exponentially higher than the known $O(\log{m})$ time complexity for dynamic exact matching without wildcards.   

Our conditional lower bound also extends to $(1+\eps)$-approximate \DIP, \DIP modulo $2$ and remarkably, to \DHD modulo $2$ with a ternary input alphabet. This latter result is in stark contrast to the complexity of \DHD modulo $2$ with a binary input alphabet which we show in Lemma~\ref{lem:dhdmod2} can be solved in $O(\log{m}/\log{\log{m}})$ query and update time. 

We complement all these conditional lower bounds with a set of unconditional lower bounds derived via reductions from different 2d-dynamic range counting problems. First we show that \DIP is at least as hard as weighted 2d-range counting. As a result, we get an unconditional lower bound of $\Omega((\log{m}/\log{\log {m}})^2)$  for \DIP. This matches the highest unconditional lower bound known for any dynamic data structure problem.  We then go on to show  $\Omega((\log^{1/2} m/\log{\log{m}})^3)$ unconditional lower bounds for  \DHD over binary alphabets, \DIP modulo $2$ over binary alphabets and \DHD modulo $2$ over ternary alphabets.  These lower bounds are derived from a recent breakthrough in the complexity of the unweighted version of 2d-range counting. To finish our unconditional lower bounds we then show $\Omega(\log{m}/\log{\log {m}})$ unconditional lower bounds for  \DHD modulo $2$ over binary alphabets, \DEM and $(1+\eps)$-approximate \DIP.   

As our final set of dynamic problems,  we move on to consider $(1+\eps)$-approximate \DHD for which we do not have matching conditional lower bounds, despite its superficial similarity to approximate \DIP.  Unlike for approximate \DIP and exact \DHD, in Section~\ref{sec:upper_approx_HD} we show markedly different upper bounds for approximate \DHD depending on whether updates may occur in only the pattern and text or in both.  For the former case we derive  $O(\eps^{-c}\; \polylog m)$ time algorithms via Johnson-Lindenstrauss sketching. The exact value of $c$ depends on the size of the input alphabet and in fact for some update operations the running time dependency on $\log{m}$ is completely removed. For the latter case with updates in both the pattern and text, our upper bound is $O(\eps^{-2} \sqrt{m}\;\polylog m)$ time.  It is an interesting and open question whether there exist matching conditional lower bounds for these versions of approximate \DHD as well. We give a summary of the results in Table~\ref{t:summary}.

\begin{table}
\begin{tabular}{|c|l|l|c|c|c|}
\hline
& Mode &Alphabet & Upper bounds & Cond. lower bounds& Uncond. lower bounds\\
\hline
\parbox[t]{2mm}{\rotatebox[origin=c]{90}{\DEM}}&exact&polynom.&$O(\sqrt{m \log m})$&$\Omega(m^{1/2 -\delta})$&$\Omega(\log m / \log \log m)$\\

\hline
\parbox[t]{2mm}{\multirow{3}{*}{\rotatebox[origin=c]{90}{\DIP}}}&exact&polynom.&$O(\sqrt{m \log m})$&$\Omega(m^{1/2 -\delta})$&$\Omega((\log m / \log \log m)^2)$\\
&$\bmod \; 2$ & $\{0,1\}$ &$O(\sqrt{m \log m})$&$\Omega(m^{1/2 -\delta})$&$\Omega((\log^{1/2} m / \log \log m)^3)$\\
&approx.&polynom.&$O(\sqrt{m \log m})$&$\Omega(m^{1/2 -\delta})$&$\Omega(\log m / \log \log m)$\\
\hline

\parbox[t]{2mm}{\multirow{3}{*}{\rotatebox[origin=c]{90}{\DHD}}}&\multirow{2}{*}{exact} & constant &$O(\sqrt{m \log m})$&$\Omega(m^{1/2 -\delta})$&$\Omega((\log^{1/2} m / \log \log m)^3)$\\
&& polynom. &$O(m^{3/4} \log^{1/2} m)$&$\Omega(m^{1/2 -\delta})$&$\Omega((\log^{1/2} m / \log \log m)^3)$\\
&\multirow{2}{*}{$\bmod \; 2$} &$\{0,1\}$ &$O(\log m / \log \log m)$& --- &$\Omega(\log m / \log \log m)$\\
& &$\{0,1,2\}$ &$O(\sqrt{m \log m})$&$\Omega(m^{1/2 -\delta})$&$\Omega((\log^{1/2} m / \log \log m)^3)$\\
\hline
\end{tabular}
\caption{Update/query time bounds for \DEM, \DHD, and \DIP for a text $T$ of length $m \le n \le 2m$ and a pattern $P$ of length $m$. For the conditional lower bounds, $\delta > 0$ is an arbitrary constant. Bounds for $(1+\eps)$-approximate \DHD are not shown (see Section~\ref{sec:upper_approx_HD} for details). }\label{t:summary}
\end{table}

\section{Related work}
In the dynamic setting we consider with single character updates, the most closely related previous work considers the problem of dynamic exact matching. In~\cite{ALLS:2007}  an $O(\log{\log{m}})$ time algorithm was shown for dynamic exact matching when updates are only permitted in the text~\cite{ALLS:2007}. In~\cite{ABR:2000} a more general data structure was developed supporting insertion and deletion of characters and movements of  arbitrary large blocks of text. This was improved in a succession of papers culminating in the work~\cite{GKKLS:2015} who give a data structure that supports, amongst other properties, concatenation, splitting and equality testing  in $O(\log{m})$ update and $O(1)$ query time. The same data structure solves, for example, the dynamic exact matching problem without wildcards problem in $O(\log{m})$ time. At the expense of $O(\log^2{m})$ updates this latter work also supports finding occurrences of a specified pattern $P$ in $O(|P|)$ time. A separate line of work has considered the static data structure problem of text indexing for approximate matching~\cite{Baeza-yates:1997a,AKLLR:2000,CGL:2004,Boytsov:2011,HHLS:2006,Boytsov:2011,ALPS:2014}.

There has also been a number of papers working on conditional hardness for other types of string problems. Larsen et al. proved lower bounds for document retrieval and forbidden pattern document retrieval conditional on hardness of boolean matrix multiplication~\cite{LMNT:2015}. Later, Kopelowitz et al. showed 3SUM-conditional lower bounds for these two problems~\cite{KPP:2015}. Backurs and Indyk~\cite{BI:SETH:2015} proved lower bounds computing the edit distance of two string. Bringmann and K{\"{u}}nnemann~\cite{Bringmann:FOCS:2015} proved lower bounds for dynamic time warping and longest common subsequence. Finally, Backurs and Indyk~\cite{backursRegular} and follow up work by Bringmann et al.~\cite{bring:Regular} proves conditional lower bounds for regular expression matching. 

\section{Upper bounds for \DHD, \DIP, and \DEM}
\label{sec:upper}
In this section we show upper bounds for \DIP, \DHD, and \DEM problems. Recall that a query $i$ asks for $f(P, T[i,\dots, i+m-1])$.  For \DIP we define $f(P, T[i,\dots,i+m-1])$ to be equal to the inner product of $P$ and $T[i,\dots,i+m-1]$, for \DHD the Hamming distance between $P$ and $T[i,\dots,i+m-1]$. In the \DEM problem we assume that $P$ and $T$ are strings over $\Sigma \cap \{?\}$, where $\Sigma$ is an integer alphabet and $?$ is a special wildcard symbol that matches any letter in $\Sigma$. We define $f(P, T[i,\dots,i+m-1])$ to be equal to zero if $P$ matches $T[i,\dots,i+m-1]$ and the number of mismatching positions otherwise. We define $n$ to be the length of the text, and $m$ to be the length of the pattern, $n \ge m$.

We will in fact present a general solution for dynamic string problems where $f$ can be represented in a particular form.  \DIP, \DHD and \DEM will seen as special cases.  The restriction is simply that $f(P, T[i,\dots,i+m-1]) = \sum_{j=1}^{j=m} g(P[j], T[i+j-1])$, where the function $g$ can be evaluated in constant time. This functional form is closely related to the idea of \emph{local} distance functions that were key to the development of fast streaming pattern matching algorithms~\cite{CEPP:2011}.  We first show that our string problems do indeed satisfy the stated requirements.

\begin{lemma}\label{lm:function}
If $f$ is inner product, Hamming distance, or exact matching with wildcards, then there exists a function $g$ such that $f(P, T[i,\dots,i+m-1]) = \sum_{j=1}^{j=m} g(P[j], T[i+j-1])$, where the function $g$ can be evaluated in constant time. 
\end{lemma}
\begin{proof}
If $f$ is inner product, we put $g(P_j, T_{i+j-1}) = P_j \cdot T_{i+j-1}$. In the case of Hamming distance, we define $g(P_j, T_{i+j-1}) = 0$ if $P_j = T_{i+j-1}$ and $g_j(P_j, T_{i+j-1}) = 1$ otherwise.

For \DEM we assume that wildcards are represented by the value $0$. It is not hard to see that we can take $g(P_j, T_{i+j-1})$ to be the characteristic function of ${(P_j-T_{i+j-1})^2 P_jT_{i+j-1} > 0}$ and indeed this observation is the basis for one of the fastest offline exact matching with wildcards algorithms~\cite{Clifford:2007}. The key property we use is that either (a) if one of $P_j$ and $T_{i+j-1}$ is a wildcard or $P_j = T_{i+j-1}$ then $g (P_j, T_{i+j-1}) = 0$, or (b) $P_j \neq T_{i+j-1}$ and then $g (P_j, T_{i+j-1}) > 0$. It follows that $f(P, T[i,\dots,i+m-1])$ equals zero if and only if $P$ and $T[i,\dots,i+m-1]$ match.
\end{proof}

We now show a solution for all dynamic string problems defined by a function $f$ that can be represented in the form above. We consider the most general update model, where we are allowed to update both the text and the pattern.

\begin{theorem}\label{th:upper}
Let $T$ be a text of length $n$, and $P$ be a pattern of length $m$. Assume $f$ can be represented as $f(P, T[i,\dots,i+m-1]) = \sum_{j = 1}^{j=m} g(P_j, T_{i+j-1})$, where $g$ can be computed in constant time, and the values $f(P, T[1,\dots,m])$, $f(P, T[2,\dots,m+1])$, \ldots, $f(P, T[n-m+1,\dots,n])$ can be computed in $\Time(n)$ time and $S(n)$ space. We can then solve the corresponding dynamic string problem in $O(\sqrt{\Time(n)})$ worst case update/query time using $O(S(n)+n)$ space. 
\end{theorem}
\begin{proof}
Let us first show a solution with $O(\sqrt{\Time(n)})$ amortised time. We start by computing values $A[1] = f(P, T[1,\dots,m]), \ldots, A[n-m+1] = f(P, T[n-m+1,\dots,n])$ in $O(\Time(n))$ time and $S(n)$ space. At all times, we maintain a list of updates $U$ that have occurred since the last moment we recomputed the values $A[i]$. Suppose that the size of $U$ is at most $\lceil \sqrt{\Time(n)} \rceil$ and a query $i$ arrives. We can then compute $A'[i] = f(P, T[i,\dots,i+m-1])$ from $A[i]$ and $U$ in the following way. We initialise $A'[i] = A[i]$, and consider each update in order. Suppose that an update change letters in a position $k$ of $P$ or $T[i,\dots,i+m-1]$, and let $P'_k$ and $T'_{i+k-1}$ be the updated letters. We remember $P'_k$ and $T'_{i+k-1}$, and set

$$A'[i] \leftarrow A'[i] - g(P_k, T_{i+k-1}) + g(P'_k, T'_{i+k-1})$$
Since $g$ can be evaluated in constant time, this step takes constant time as well. Therefore, the time to perform each query is $O(\sqrt{\Time(n)})$. When the size of $U$ reaches $\lceil \sqrt{\Time(n)} \rceil$, we apply the updates in $U$ to $T$ and $P$, empty $U$, and recompute the values $A[i]$ from scratch. The amortised cost of an update is therefore $O(\sqrt{\Time(n)})$.

We can de-amortise the solution in a standard way. Namely, we restart the computation of the values $A[i]$ each $\lceil \sqrt{\Time(n)}/2 \rceil$ updates, and run $\Theta(\sqrt{\Time(n)})$ steps of the computation per each of the $\lceil \sqrt{\Time(n)}/2 \rceil$ subsequent updates. While the computation is not over, we make use of the previously computed values $f(P, T[1,\dots,m]) ,\ldots, f(P, T[n-m+1,\dots,n])$ to answer queries. As before, we will need to correct the value of the function $g$ in at most $\lceil \sqrt{\Time(n)}/2 \rceil$ positions. Note that apart from the space we need for computing the values $f(P, T[1,\dots,m])$, $f(P, T[2,\dots,m+1])$, \ldots, $f(P, T[n-m+1,\dots,n])$, we need only $O(n)$ space.
\end{proof} 

\begin{lemma}\label{cor:upper_HD}
For a text $T$ of length $m \le n \le 2m$, and a pattern $P$ of length $m$, problem \DHD can be solved in  $O(\sqrt{m \log m})$ query/update time for constant-size alphabets, and in $O(m^{3/4} \log^{1/4} m)$ query/update time for polynomial-size alphabets. Both solutions use $O(m)$ space, and both updates to the text and to the pattern are allowed.
\end{lemma}
\begin{proof}
If the alphabet is binary, the values $f(P, T[1,\dots,m])$,\ldots, $f(P, T[n-m+1,\dots,n])$ can be computed by running the FFT algorithm twice. Recall that the FFT algorithm computes the inner product for each alignment of two strings. By running the FFT algorithm on $P$ and $T$ for the first time, we obtain, for each $i$, the number of positions $j$ such that $P[j] = T[i+j] = 1$. By running it for the second time on the copies $P$ and $T$ where each bit is flipped, we obtain, for each $i$, the number of positions $j$ such that $P[j] = T[i+j] = 0$. We can then compute the values $f(P, T[1,\dots,m]) ,\ldots, f(P, T[n-m+1,\dots,n])$ in linear time. For this algorithm, $\Time(n) = O(n \log n) = O(m \log m)$. For alphabets of constant size $|\Sigma|$, we run the FFT algorithm $|\Sigma|$ times, once for each letter $a \in \Sigma$, on the copies of $P$ and $T$ where $a$ is replaced with $1$ and all letters in $\Sigma \setminus \{a\}$ are replaced with $0$. $T(n) = O(m \log m)$ as well. For polynomial-size alphabets, $\Time(n) = O(n\sqrt{n \log{n}}) = O(m \sqrt{m \log m})$ and $S(n) = O(n) = O(m)$ bounds were shown independently by Abrahamson~\cite{Abrahamson:1987} and Kosaraju~\cite{Kosaraju:1987} in 1987. The claim immediately follows from Lemma~\ref{lm:function} and Theorem~\ref{th:upper}.
\end{proof}

\begin{lemma}\label{cor:upper_IP_EM}
For a text $T$ of length $m \le n \le 2m$, and a pattern $P$ of length $m$, problems \DIP and \DEM can be solved in $O(\sqrt{m \log m})$ query/update time using $O(m)$ space. Both updates to the text and to the pattern are allowed.
\end{lemma}
\begin{proof}
For both problems, $\Time(n) = O(n \log n) = O(m \log m)$ and $S(n) = O(n) = O(m)$. For inner product, this is a direct corollary of the FFT algorithm. The bound for exact matching with wildcards was demonstrated in~\cite{CH:2002,Clifford:2007}. The claim follows from Lemma~\ref{lm:function} and Theorem~\ref{th:upper}.
\end{proof}

We now extend our solution to a general value of $n$. In this case there is also an additional cost of computing the full set of solutions before the first query or update is performed which we omit from the following theorem.

\begin{theorem}\label{thm:upperbounds-exact}
For a text $T$ of length $n \ge m$, and a pattern $P$ of length $m$, there is a linear-space data structure that solves
\begin{enumerate}[(a)]
\item the \DHD problem in $O(\sqrt{m \log m})$ query/update time for constant-size alphabets, and in $O(m^{3/4} \log^{1/4} m)$ query/update time for polynomial-size alphabets, and the \DIP and the \DEM problems in $O(\sqrt{m \log m})$ query/update time if only updates to the text are allowed;
\item the \DHD problem in $O(\frac{n}{m} \cdot \sqrt{m \log m})$ query/update time for constant-size alphabets, and in $O(\frac{n}{m} \cdot m^{3/4} \log^{1/4} m)$ query/update time for polynomial-size alphabets, and the \DIP and the \DEM problems in $O(\frac{n}{m} \sqrt{m \log m})$ query/update time when updates are allowed both to the text and to the pattern.
\end{enumerate}
\end{theorem}
\begin{proof}
We first partition $T$ into blocks of length $2m$ overlapping by $m$ positions (the last block may be shorter). Note that for each $i$ a string $T[i,\dots,i+m-1]$ is a substring of one of such blocks, and each position of $T$ belongs to at most two blocks. Suppose that we have a solution for a text of length $2m$ and a pattern of length $m$ with update time $t_{u}$, query time $t_q$, and space $S$.
\begin{enumerate}[(a)]
\item If only updates to the text are allowed, we can apply this solution independently to each of the blocks. Note that an update of the text changes at most two blocks, and therefore we obtain a solution for $T$ with update time $O(t_{u})$, query time $t_q$, and space $O(\frac{n}{m} \cdot S)$. 
\item If updates are allowed both to the text and to the pattern, we obtain a solution for $T$ with update time $O(\frac{n}{m} \cdot t_{u})$, query time $t_q$, and space $O(\frac{n}{m} \cdot S)$. 
\end{enumerate}
The claim follows from Lemmas~\ref{cor:upper_HD} and~\ref{cor:upper_IP_EM}.
\end{proof}

\section{Upper bounds for dynamic approximate Hamming distance}
\label{sec:upper_approx_HD}
In this section we develop algorithms for an approximate version of \DHD. We will refer to this version as $\DaHD$. In this problem a query $i$ must return a $(1+\eps)$-approximation of the Hamming distance between $P$ and $T[i,\dots,i+m-1]$, where $\eps > 0$ is a parameter of the algorithm.  Unlike the other problems we have considered, the complexity of \DaHD appears to have a strong dependence on whether updates are permitted only in the pattern or text or in both.  At one extreme, when updates are only permitted in the pattern and the input alphabet is binary, we show in Theorem~\ref{th:upper_DaHD} a data structure that takes  $O(1/\eps)$ update and $O(1/\eps^2)$ query time.  However if updates can occur in both the pattern and the text, then the complexity increases dramatically to be at least that of exact \DIP, \DEM and \DHD over binary alphabets.

In Section~\ref{sec:upper} we showed that the \DHD problem can be solved in $O(m^{1/2} \log^{1/2} m)$ query/update time for constant-size alphabets, and in $O(m^{3/4} \log^{1/4} m)$ query/update time for polynomial-size alphabets.  We start our exploration of the complexity of \DaHD by showing that this dependence on the alphabet size is almost completely removed in this approximate setting.  The solution we give is deterministic and is based on the mapping idea of Karloff~\cite{Karloff:1993}. 

\begin{lemma}[\cite{Karloff:1993}]
Let $\Sigma$ be the alphabet of $P$ and $T$. There exists $\Theta((1/\eps^2) \log^2 n)$ deterministic mappings $map_j : \Sigma \rightarrow \{0, 1\}$ such that a $(1 +\eps)$-approximation of the Hamming distance between $P$ and $T$ at a particular alignment can be given by a normalised average of the Hamming distances between $map_j(P) = map_j(P_1) \ldots map_j(P_n)$ and $map_j(T) = map_j(T_1) \ldots map_j(T_n)$ at this alignment. Each mapping can be stored as a look-up table that permits to compute each $map_j(P_k)$ or $map_j(T_k)$ in $O(1)$ time. 
\end{lemma}

\begin{corollary}\label{cor:upper_HD_approx}
For a text $T$ of length $m \le n \le 2m$, and a pattern $P$ of length $m$, the \DaHD problem over polynomial-size alphabets can be solved in $O((1/\eps^2) \sqrt{m} \cdot \polylog m)$ query/update time and $O((1/\eps^2) m \log^2 m)$ space.
\end{corollary}
\begin{proof}
We consider Karloff's mappings $map_j$. For each $j$, we run our $\DHD$ solution for constant-size alphabets (Lemma~\ref{cor:upper_HD}) on $map_j(P)$ and $map_j(T)$. The claim immediately follows.
\end{proof}

We now present several randomised solutions for \DaHD in two special update models where we are allowed to update either only the text or only the pattern. 
We first assume a binary input alphabet, and then show how to extend our solutions to constant-size and then later polynomial-size alphabets as well.

\begin{theorem}\label{th:upper_DaHD}
For a text $T$ of length $n \ge m$, and a pattern $P$ of length $m$, there is a randomised data structure for the \DaHD problem over a constant-sized alphabet with
\begin{enumerate}[(a)]
\item $O(1/\eps)$ update time, $O(1/\eps^2)$ query time, and $O((1/\eps^2) \cdot n)$ space if only updates to the pattern are allowed;
\item $O((1/\eps) \cdot \polylog n)$ update time and $O((1/\eps^2) \cdot \polylog n)$ query time using $O((1/\eps^2) \cdot n \; \polylog n)$ space if only updates to the text are allowed.
\end{enumerate}
Each answer is correct with constant probability.
\end{theorem}
\begin{proof}
Let us first assume the input alphabet is of constant size.  We will make use of the sparse Johnson-Lindenstrauss transform by Kane and Nelson~\cite{KN:2014} defined by a random $\Theta(1/\eps^{2}) \times n$ matrix $M$ such that its entries are from $\{-1, 0, 1\}$, and each of its columns contains $s = \Theta(1/\eps)$ non-zero entries. The result of a transform, which we call a sketch, is defined to be equal to $ s^{-1/2} M \cdot x$. Kane and Nelson showed how to choose a distribution on such matrices such that, with constant probability, the square of the $L_2$ norm of the difference of the sketches of two strings gives a $(1+\eps)$-approximation of  Hamming distance. 

\begin{enumerate}[(a)]
\item During the preprocessing step we compute the sketch of $P$ and of each $m$-length substring of $T$. When an update to $P$ arrives, we update its sketch in a naive way in $O(1/\eps)$ time. When a query $i$ arrives, we can compute a $(1+\eps)$-approximation of the Hamming distance between $P$ and $T$ by computing the  $L_2$ norm of the difference of the sketches of $P$ and $T[i,\dots,i+m-1]$. Since the sketches are the vectors of length $1/\eps^2$, this can be done in $O(1/\eps^2)$ time.

\item For this model, we will need a sketch that gives $(1+\eps)$-approximation of  Hamming distance with error probability $\Theta(1/\log m)$. This can be achieved by repeating the scheme $\Theta(\log \log m)$ times. During the preprocessing, we first compute  $\Theta(\log \log m)$ sketches for each $2^k$-length substring of the pattern $P$, where $k = 1,2,\dots,\log m$. We then compute  $\Theta(\log \log m)$ sketches for each substring $T[i \cdot 2^k+1, \dots, (i+1) \cdot 2^k]$. We call such substrings of $T$ canonical. When an update $(i, \sigma)$ arrives, we need to fix the sketches of $O(\log m)$ canonical substrings (since $T_i$ belongs to $O(\log m)$ such substrings), which can be done in $O((1/\eps) \log m \log \log m)$ time.  A query $i$ can be answered in $O((1/\eps^2) \log m \log \log m)$ time: First, we partition $T[i, \dots, i+m-1]$ into $O(\log m)$ canonical substrings $S_1, \ldots, S_k$. Secondly, we compute a $(1+\eps)$-approximation of the Hamming distance between each $S_i$ and the corresponding substring of $P$ using the sketches. Finally, we sum up all approximations to obtain the answer. Since the probability to error on each pair of substrings is $\Theta(1/\log m)$, the total error probability is constant by the union bound.
\end{enumerate}

Both algorithms can be extended to work for any constant sized alphabet by expanding the input alphabet in unary. That is we replace the letter $i$ with a binary vector $0\dots010\dots 0$, where the set bit is in the $i$-th position.  
\end{proof}

\begin{corollary}
For a text $T$ of length $n \ge m$, and a pattern $P$ of length $m$, and $\eps > 1/n$, there is a randomised data structure for the \DaHD problem over polynomial-size alphabets with 
\begin{enumerate}[(a)]
\item $O((1/\eps^3) \cdot \polylog n)$ update time, $O((1/\eps^4) \cdot \polylog n)$ query time, and $O((1/\eps^4) \cdot n \; \polylog n)$ space if only updates to the pattern are allowed;
\item $O((1/\eps^4) \cdot \polylog n)$ update time, $O((1/\eps^4) \cdot \polylog n)$ query time, and $O((1/\eps^4) \cdot n \; \polylog n)$ space if only updates to the text are allowed.
\end{enumerate}
Each answer is correct with constant probability.
\end{corollary}
\begin{proof}
We reduce the alphabet to binary by applying Karloff's mappings. There are $\Theta((1/\eps^2) \log^2 n)$ mappings, and to compute the Hamming distance between $P$ and $T[i,\dots,i+m-1]$ we need to compute the Hamming distance for each pair $map_j(P)$ and $map_j(T[i,\dots,i+m-1])$. To achieve constant error probability, we run $\Theta(\log ((1/\eps) \log n)) = \polylog n$ instances of the algorithm for text-only or pattern-only updates (Theorem~\ref{th:upper_DaHD}). (We note that we will achieve $(1+\eps)^2$-approximation, which is $(1+\eps')$-approximation for $\eps' = 2\eps + \eps^2$.)
\end{proof}

\section{Lower bounds}
In this section we demonstrate conditional and unconditional lower bounds for different variants of \DEM, \DIP, and \DHD. The conditional lower bounds are derived from the hardness of a well-known problem, online Boolean matrix-vector product (\OMv). Fig.~\ref{fig:reductions} summarises the reductions we use. 

\begin{figure}[h!]
\begin{tikzpicture}[node distance=2cm,xscale=1.15,yscale=0.8]
\node (A) at (5,4) {\OMv};
\node[text width=3.5cm]  (B) at (0,2) {\DHD modulo 2 (ternary alphabet)};
\node (C) at (0,0) {\DHD};
\node[text width=3.4cm] (D) at (3.5,2) {\DIP modulo 2 (binary alphabet)};
\node (E) at (6.5,2) {$(1+\eps)$-approx. \DIP};
\node (F) at (5,0) {\DIP};
\node (G) at (10,2) {\DEM};

\draw[->] (A) -- (D);
\draw[->] (A) -- (E);
\draw[->] (A) -- (9.7,2.4);
\draw[->] (B) -- (C);
\draw[->] (D) -- (F);
\draw[->] (E) -- (F);
\draw[->] (F) -- (C);
\draw[->] (1.9,2) -- (1.1,2);
\end{tikzpicture}
\caption{Reductions between \OMv and different variants of \DEM, \DIP, and \DHD.}
\label{fig:reductions}
\end{figure}
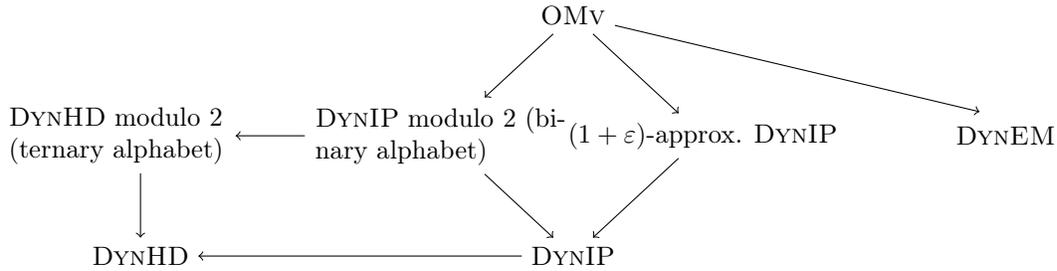	

	\subsection{Reductions between \DIP, \DHD and \DHD modulo $2$}	
	Before we get to our main lower bounds results we will first establish the relationship between some of the dynamic string problems we consider.

\begin{lemma}\label{lem:dhdfromdip}
\DHD is at least as hard as \DIP over binary alphabets.
\end{lemma}
\begin{proof}
 We map the input alphabet of the text and the pattern separately. Take an instance of \DIP where the input alphabet is binary. In order to transform it into an instance of \DHD each $1$ in the pattern or text is mapped to the string $111$ in the \DHD instance. Similarly, a $0$ in the pattern is mapped to the string $010$ and a $0$ in the text is mapped to the string $100$.  This transformation ensures that any two symbols that align in the \DIP instance will give Hamming distance $2$ in the \DHD instance except when two $1$s align. In this case the Hamming distance will be $0$. We can therefore infer the inner product from the Hamming distance: The inner product will be equal to the length of the pattern minus the Hamming distance divided by two.
\end{proof}

We will later show both conditional and unconditional lower bounds not only for \DIP but also for \DIP modulo $2$. The following two lemmas will lead to perhaps our most surprising result which is that \DHD modulo $2$ over \emph{ternary} alphabets is exponentially harder to solve than \DHD modulo $2$ over a binary alphabet.  It is worth emphasising by way of contrast that in the standard offline pattern matching setting, the asymptotic complexity of computing the Hamming distance at all alignments of a pattern and text is identical for any constant sized input alphabet. 

\begin{lemma}\label{lem:dhdfromdip-modulo}
\DHD modulo $2$ over a ternary alphabet is at least as hard as \DIP modulo $2$ over a binary alphabet.
\end{lemma}
\begin{proof}
We again map the input alphabet of the text and pattern separately. Take an instance of \DIP modulo $2$ where the input alphabet is binary.  Each $1$ in the pattern is mapped to the string $22$ and each $0$ in the pattern is mapped to the string $01$.  Each $1$ in the text is mapped to the string $11$ and each $0$ in the text is mapped to the string $02$. This transformation ensures that any two symbols that align in the \DIP  modulo $2$ instance will give Hamming distance $1$ in the \DHD modulo $2$ instance except for when two $1$s align in the \DIP modulo $2$ instance when the resulting Hamming distance is $2$. Therefore, the inner product modulo $2$ is equal to the length of the pattern minus the Hamming distance modulo $2$.
\end{proof}

However, \DHD modulo $2$ over a binary alphabet is much easier than \DHD modulo $2$ over a ternary alphabet.

\begin{lemma}\label{lem:dhdmod2}
For a binary text $T$ of length $n \ge m$, and a binary pattern $P$ of length $m$ the \DHD modulo $2$ problem can be solved in $O(\log{m}/\log{\log{m}})$ update/query time using $O(n)$ space. There is a matching unconditional lower bound for update/query time as well.
\end{lemma}
\begin{proof}
As before, we divide the text $T$ into $2m$-length blocks overlapping by $m$ positions. We will show that for each block \DHD modulo $2$ can be solved in $O(\log{m}/\log{\log{m}})$ update/query time using $O(m)$ space, hence giving the claim.

Consider a $2m$-length block of $T$. In order to answer a query at alignment $i$ for \DHD modulo $2$ we need only to sum, modulo $2$, the number of $1$s in the pattern and the corresponding substring of the text $T[i,\dots,i+m-1]$.  This can be seen via a simple proof by induction as follows. As the base case consider two strings of length $1$ and let all arithmetic be over $\mathbb{Z}_2$.  In this case the Hamming distance is the sum of the Hamming weights of the two strings. For the inductive step, extend each of these two strings by one bit and observe that the new Hamming distance is the old Hamming distance before extending the strings plus the sum of the two new bits over $\mathbb{Z}_2$.   

The Hamming weight of the pattern can be maintained straightforwardly. We argue that answering queries for the Hamming weight of substrings of the block is equivalent to the prefix sum problem modulo $2$. To reduce from this problem to prefix sum we need only observe that we can compute the number of $1$s in $T[i, \dots, i+m-1]$ by subtracting the prefix sum up to index $i-1$ from the prefix sum up to index $i+m-1$. To reduce from prefix sum to the \DHD modulo $2$ problem we construct a text of length $2m$ with the first half all zeros and the second half as a copy of the prefix sum array. Setting the pattern to all $1$s we can compute the prefix sum modulo $2$ up to index $i$ of its array of length $m$ by performing a query at index $i$ of the text. It follows from the upper and lower bounds of~\cite{PD2004:Partial-sums} that the complexity of \DHD modulo $2$ over a binary alphabet is  $\Theta(\log{m}/\log{\log{m}})$. 
\end{proof}
	
	\subsection{Conditional lower bounds}
	We will now give lower bounds for our dynamic string problems conditional on the hardness of a well known problem. The \OMv problem was introduced in~\cite{HKNS:2015} as a means to prove conditional lower bounds for a number of dynamic problems. In this problem we are first given an $r \times r$ Boolean matrix $M$. We then receive $r$ vectors $v_1, \ldots, v_r$, one by one. After seeing each vector $v_i$, we have to output the product $M v_i$ (over the Boolean semi-ring) before we receive the next vector. A naive algorithm can solve this problem using $O(r^3)$ time in total with the current fastest solution taking $O(r^3/2^{\Omega(\sqrt{\log r})})$ time~\cite{LW:2017}. The \OMv conjecture is as follows: 

\begin{conjecture}[\OMv Conjecture~\cite{HKNS:2015}]\label{conj:omv_conjecture}
For any constant $\epsilon > 0$, there is no $O(r^{3-\epsilon})$-time algorithm that solves the \OMv problem with error probability of at most $1/3$.
\end{conjecture}

\begin{theorem}\label{thm:conditional-main}
Assuming the \OMv conjecture, there does not exist an algorithm running in $O(m^{1/2-\epsilon})$ for the maximum of query and update time for \DEM, \DIP, and \DHD. The same lower bound holds for \DIP modulo $2$, for $(1+\eps)$-approximate \DIP, and for \DHD modulo $2$ over ternary alphabets. The same lower bound holds even when updates are permitted only in the pattern or only in the text.
\end{theorem}
\begin{proof}
We first give a reduction from the online Boolean matrix-vector multiplication problem to \DEM.  We create a text $T$ of length $2m = 2r^2$ from the matrix $M$ by concatenating the $r$ rows of $M$ one after another and filling the rest of $T$ with the symbol $1$ repeated $r^2$ times. Now consider a single Boolean matrix vector product $Mv_i$. The pattern $P$ has length $m = r^2$. Its first $r$ symbols are a copy of the vector $v_i$ but with all $0$s replaced by the wildcard symbol ? and all $1$s replaced by the symbol $0$. The remaining $r^2-r$ symbols are set to the wildcard symbol ?. To perform a Boolean matrix vector multiplication we perform $m$ exact match with wildcard queries at indices $1,r+1,2r+1,\dots,(r-1)r+1$. If a query $i$ returns a match then $Mv_i[j]=0$ and $Mv_i[j]=1$ otherwise. If follows that any algorithm for \DEM running in $O(m^{1/2-\eps})$ for the maximum of query and update time implies an $O(r^{3-\eps})$-time algorithm that solves the online Boolean matrix-vector multiplication problem, thereby contradicting the OMv conjecture.

\DIP and \DHD are at least as hard as \DIP modulo $2$, so it suffices to show the lower bound for the latter. We give a similar reduction from \OMv but this time with an extra randomisation step. We create a text $T$ of length $2m = 2r^2$ from the matrix $M$ by concatenating the $r$ rows of $M$ one after another and filling the rest of $T$ with the symbol $0$ repeated $r^2$ times. Now consider a single Boolean matrix vector product $Mv_i$.  We create a pattern $P$ of length $m = r^2$ with the first $r$ symbols being a copy of $v_i$ and the remaining $r^2-r$ symbols set to $0$. We now flip each set bit in $P$ with probability $1/2$ and compute inner product modulo $2$ queries at  indices $1,r+1,2r+1,\dots,(r-1)r+1$. If $Mv_i[j] = 0$ then an inner product query $j$ will always return $0$. If $Mv_i[j]=1$ then the inner product query will return $1$ with probability $1/2$.  This gives a probability of at least $1/2$ of giving the correct answer for each $Mv_i[j]$. We amplify the probabilities by repeating the randomised procedure $O(\log m)$ times using the fact that we have one-sided error at each iteration. It then follows that there does not exist an algorithm running in $O(m^{1/2-\eps})$ for the maximum of query and update time for \DIP modulo $2$ unless the \OMv conjecture is false. 

The lower bound for $(1+\eps)$-approximate \DIP follows from the same reduction with the arithmetic performed over the reals rather than modulo $2$ and without the randomisation step. This is because a $(1+\eps)$-approximation must be able to distinguish zero and non-zero inner products which is sufficient for our reduction from OMv.    

The lower bound for \DHD modulo $2$ over a ternary alphabet now follows from Lemma~\ref{lem:dhdfromdip-modulo}. 

If updates are only allowed in the text then we derive the same lower bound as before by modifying our reductions. Let us take the reduction from the online Boolean matrix-vector multiplication problem to \DIP modulo $2$ as an example. The other lower bounds follow analogously. We create a pattern $P$ of length $m = r^2$ from the matrix $M$ by concatenating the $r$ rows of $M$ one after another. The text is of length $2m = 2r^2$ and will be all $0$s except for the substring $T[r^2-r+1,\dots, r^2]$. In order to perform a single Boolean matrix vector product $Mv_i$ the substring is updated so that $T[r^2-r+1, \dots, r^2] = v_i$ and we then flip each set bit in $T$ with probability $1/2$. We then compute inner product queries modulo $2$ at indices $1,r+1,2r+1,\dots,(r-1)r+1$ which give the correct answer for each query with probability at least $1/2$. We can amplify the probability as before giving us the desired lower bound.   
\end{proof}

Our lower bound also holds for \DIP modulo $c$ for any $c \geq 2$.

\begin{corollary}\label{cor:dipmodc}
Let integer $c \geq 2$. Assuming the \OMv conjecture, there does not exist an algorithm running in $O(m^{1/2-\epsilon})$ for the maximum of query and update time for \DIP modulo $c$.
\end{corollary}
\begin{proof}
Let the input alphabet be binary as before and perform the same randomised reduction from OMv as in the proof of Theorem~\ref{thm:conditional-main}. If the inner product equals $0$ then we always give the correct answer. If the inner product is greater than $0$ then after flipping the set bits, the inner product modulo $c$ is greater than $0$ with probability that tends asymptotically to $\frac{c-1}{c}$.  We can then amplify the probabilities to ensure that every value in the matrix-vector product is correct with constant probability as before.
\end{proof}
	
	\subsection{Unconditional lower bounds}
	In this section we will give unconditional lower bounds for all the problems we have considered except \DaHD.  Although these bounds are necessarily much lower than the conditional lower bounds we gave previously, they nonetheless match in many cases the limits of what is known unconditionally for any dynamic data structure.

We first show lower bounds for the \DIP and the \DHD problems by reduction from the dynamic weighted range counting problem. In this problem, we are given a $r \times r$ grid $D$. The points in the grid are assigned integer weights, and at any moment there can be at most $r$ non-zero weights $w_i$. For our problem $r = m^{1/3}$. Updates may change the weight of a point and a query $(i,j)$ asks for $\sum_{x\leq i, y\leq j} D_{x,y}$. In~\cite{Larsen:2012} Larsen gave an $\Omega((\log{r}/\log{\log {r}})^2)$ lower bound for the maximum of query and update time for dynamic weighted range counting. This lower bound does not hold however in the unweighted case (where the weights are in $\{0,1\}$) and giving an $\omega(\log{r})$ lower bound for this situation remained an important open problem for a number of years. Recently in~\cite{larsen:2017:crossing} a new $\Omega((\log^{1/2} r/\log{\log{r}})^3)$ lower bound was given for this unweighted range counting problem which also holds over $\mathbb{F}_2$. 

\begin{theorem}\label{thm:uc-dip}
The \DIP problem has an unconditional $\Omega((\log{m}/\log{\log {m}})^2)$ lower bound for the maximum of query and update time for polynomial-size alphabets. \DHD over binary alphabets, \DIP modulo $2$ over binary alphabets and \DHD modulo $2$ over ternary alphabets have an $\Omega((\log^{1/2} m/\log{\log{m}})^3)$ lower bound.  
\end{theorem}
\begin{proof}
We give a reduction from dynamic range counting to \DIP. We take an instance of the problem for $r = m^{1/3}$ and create a text $T$ of length $2m$ and a pattern $P$ of length $m$. The text has all symbols set to $0$ except $T_{m-m^{1/3}+1}$,\dots, $T_{m}$ that are set to $w_1$, \dots, $w_{m^{1/3}}$ respectively. For each of the $m^{2/3}$ different possible queries to $D$, a subset of the $w_i$'s will be included in the query. We create a pattern $P$ so that $P_{j m^{1/3}+i-1} = 1$ if weight $w_i$ is included in the range for query $j$ and $P_{jm^{1/3}+i-1} =0$ otherwise. 

To perform a range counting query, we need to align the relevant substring of the pattern of length $m^{1/3}$ with $T[m-m^{1/3}+1,\dots, m]$ and perform an inner product query. Our lower bounds then follow from the lower bounds for the weighted and $\mathbb{F}_2$ versions of dynamic range counting and Lemmas~\ref{lem:dhdfromdip} and~\ref{lem:dhdfromdip-modulo}.
\end{proof}

Finally, we give lower bounds for the \DEM and the $(1+\eps)$-approximate \DIP problems by reduction from the dynamic range emptiness problem.
In this problem, the set-up is exactly like in the unweighted dynamic range counting problem above, and a query $(i,j)$ asks if $\sum_{x\leq i, y\leq j} D_{x,y} = 0$. In~\cite{AHR:1998}, Alstrup et al. showed a $\Omega(\log{r}/\log{\log{r}})$ lower bound for this problem.

\begin{theorem}
\DEM and $(1+\eps)$-approximate \DIP have unconditional $\Omega(\log{m}/\log{\log {m}})$ lower bounds for the maximum of query and update time. 
\end{theorem}
\begin{proof}
Consider an instance of two dimensional range emptiness on $D$ for $r = m^{1/3}$. We take an instance of this problem and create a text $T$ of length $2m$ and a pattern $P$ of length $m$. The text has all values set to $0$ except $T_{m-m^{1/3}+1}$, \dots, $T_{m}$ set to $w_1$,\dots,$w_{m^{1/3}}$ respectively. For each of the $m^{2/3}$ different possible queries to $D$ in the dynamic range emptiness problem, a subset of the $w_i$'s will be included in the query. We create a pattern $P$ so that $P_{jn^{1/3}+i-1} = 0$ if weight $w_i$ is included in the range for query $j$ and $P_{jn^{1/3}+i-1} = \; ?$ otherwise. If an exact match with wildcards query returns True then we know that all the weights in the corresponding range are 0. If it returns False then we know the range is not empty. We therefore have reduced from two dimensional range emptiness to \DEM giving an $\Omega(\log{m}/\log{\log{m}})$ lower bound for \DEM.

For the $(1+\eps)$-approximate dynamic inner product problem we must be able to distinguish an inner product of zero from all other values. We therefore use the same reduction from the proof of Theorem~\ref{thm:uc-dip} but this time only report whether the approximate inner product is greater than zero. The result of this query is sufficient to determine the answer to a range emptiness query and we therefore derive the same  $\Omega(\log{m}/\log{\log{m}})$ lower bound.
\end{proof}

\bibliography{bibliofile}
\end{document}